\documentclass[10pt,twocolumn,twoside]{IEEEtran}
\usepackage{cite}
\usepackage{amsthm}
\usepackage{amssymb}
\usepackage{amsmath,amsfonts}
\usepackage{algorithmic}
\usepackage{array}
\usepackage[caption=false,font=normalsize,labelfont=sf,textfont=sf]{subfig}
\usepackage{textcomp}
\usepackage{stfloats}
\usepackage{url}
\usepackage{verbatim}
\usepackage{graphicx}
\def\BibTeX{{\rm B\kern-.05em{\sc i\kern-.025em b}\kern-.08em
    T\kern-.1667em\lower.7ex\hbox{E}\kern-.125emX}}
\usepackage{balance}

\newtheorem{lemma}{Lemma}
\newtheorem{theorem}{Theorem}
\newtheorem{definition}{Definition}
\newtheorem{remark}{Remark}
\newtheorem{assumption}{Assumption}

\newtheorem{example}{Example}

\begin{document}
\title{Autonomous Synchronization of Discrete-Time Heterogeneous Multiagent Systems}
\author{Wei Hu,~\IEEEmembership{Quanyi Liang}
\thanks{This work has been submitted to the IEEE for possible publication. Copyright may be transferred without notice, after which this version may no longer be accessible. Wei Hu (first author) and Quanyi Liang (corresponding author) are with the School of Mathematical Sciences, Beihang University, Beijing 100191, China  
(e-mail: 20377116@buaa.edu.cn; qyliang@buaa.edu.cn).}}

\markboth{}%
{}

\maketitle

\begin{abstract}
    This paper investigates the autonomous synchronization problem for discrete-time heterogeneous multiagent systems. 
    The synchronization problem is transformed into the asymptotic decoupling problem of stable modes in a class of discrete-time linear time-varying systems, 
    for which we provide a sufficient condition. 
    Leveraging this condition, synchronization conditions are established. 
    The synchronization conditions are based on the average of the agents' initial dynamic matrices, 
    without requiring the differences among these matrices to be small. 
    This approach reduces the conservativeness of existing conditions and achieves a unification of both homogeneous and heterogeneous systems. 
    Numerical simulation results are provided to support the theoretical findings.

\end{abstract}


\IEEEpeerreviewmaketitle

\section{Introduction}
\label{sec:introduction}

\subsection*{1. Problem Background and Literature Review}
Over the past several decades, the synchronization of multiagent systems (MASs) has become an important and mature area of research, 
driven by its broad applications across various domains, 
including biological rhythms in animals, robotic formation control, power grid systems, 
and artificial robot coordination \cite{Wah2007, Kocarev2013, Chen2014, Motter2013, Rahimi2014}.
These studies primarily aim to create cooperative control algorithms 
that enable a group of agents, subsystems, or signal generators to synchronize 
by forcing their controlled dynamics and trajectories to converge.

In synchronization of MASs, early studies focused on homogeneous systems \cite{olfati2004consensus, ren2005second, scardovi2008synchronization}. 
In particular, \cite{2} proposed a distributed control strategy of homogeneous systems, achieving exponential synchronization under certain conditions.
Having established fundamental synchronization results for homogeneous systems, 
a growing research effort has shifted toward the more challenging and realistic scenario of heterogeneous systems\cite{liu2024, sun2023}.
However, the synchronization of heterogeneous systems often rely on virtual leaders \cite{WIELAND20111068, kim2010output}.

In order to remove the reliance on virtual leaders for synchronization of heterogeneous systems, 
many studies have focused on the autonomous synchronization problem \cite{5},
which emphasizes that the synchronized dynamics and states are not pre-specified 
but autonomously emerge from the agents' inherent properties, offering greater practical relevance.

To achieve the autonomous synchronization of heterogeneous MASs, it is assumed in\cite{5} that 
the agent dynamics $S_i(t)$ converge exponentially to the average $S_{\infty} \triangleq \frac{1}{N}\sum_{i=1}^N S_i(0)$. 
Furthermore, it is shown that synchronization can be achieved when this convergence rate exceeds the growth rate of the unstable modes of the agents,
which serves as the sufficient condition for synchronization. 
This growth rate is characterized by the maximum real part of the eigenvalues of $S_{\infty}$ for continuous-time systems, 
or by the operator norm of $S_{\infty}$ for discrete-time systems.
It should be noted that, in the discrete-time case, 
the proper characterization for the growth rate of agents' unstable modes is the spectral radius $\rho(S_{\infty})$, 
rather than the operator norm $\|S_{\infty}\|$\cite{agarwal2000difference}.
Consequently, a relaxed synchronization condition is required for the discrete-time case, which should be expressed in terms of the spectral radius.

Furthermore, due to the stringent constraints imposed on the agents' initial dynamics $S_i(0)$,
the results of \cite{5} are effectively applicable only when 
the heterogeneity of MASs is sufficiently small, 
i.e., the differences among $S_i(0)$ are sufficiently small.
Motivated by this limitation, the authors of \cite{5} proposed two conjectures in a separate study\cite{yan2022}: 
for general heterogeneous systems, 
fully distributed autonomous synchronization of heterogeneous MASs can be achieved---one for the continuous-time case and one for the discrete-time case.
Recent study\cite{yan2025} has solved the conjecture of continuous-time case.
This inspires us to solve the problem of autonomous synchronization for general discrete-time heterogeneous MASs.

In addition, under a suitable coordinate transformation, heterogeneous multiagent systems $\xi_i(t+1) = S_i(t) \xi_i(t) + B u_i(t)$ 
($\xi_i $ is the state of agent $i$)
can be converted into a linear time-varying (LTV) system as follows:
\begin{equation}\label{eq:main-system}
X_{t+1} = (A + P_t) X_t.
\end{equation}
And achieving asymptotic decoupling of system \eqref{eq:main-system} is a key problem in synchronization of heterogeneous MASs.
For the continuous-time case, under the assumption that $\|P_t\|$ converges to $0$  exponentially,
\cite{1} has provided
the sufficient condition and convergence rate estimates for asymptotic decoupling of systems \eqref{eq:main-system}.
For the discrete-time case, a sufficient condition for asymptotic decoupling remains to be established,
which could technically facilitate the investigation of heterogeneous MASs.

\subsection*{2. Contributions of This Paper}
This paper aims to address these issues.
First, under a weaker assumption that $\|P_t\|$ converges to $0$(not necessarily exponentially\cite{1}), 
we provide the same sufficient condition for asymptotic decoupling of stable modes in discrete-time system (Theorem \ref{thm1}).
The proof leverages coordinate transformations to eliminate the influence of Jordan blocks in $A^s$ while preserving asymptotic equivalence.

Second, applying the above theoretical result to the discrete-time autonomous synchronization of heterogeneous MASs, 
we derive synchronization conditions (Theorem \ref{thm2}). 
The key to proving Theorem \ref{thm2} lies in choosing an appropriate coupling strength for the synchronization of agents' dynamics.
With this choice of coupling strength, 
the dynamics matrices synchronize to the average dynamics $S_{\infty}$ at a rate faster than $1/\rho(S_{\infty})$.
Consequently, the corresponding system \eqref{eq:main-system} achieves asymptotic decoupling, 
and thus the multiagent systems achieves synchronization.
And the synchronization conditions depend on the average of the agents' dynamics matrices $S_\infty$,
making our results applicable to multiagent systems with large heterogeneity.
Therefore, while \cite{yan2025} solved the continuous-time conjecture, 
our work makes a significant progress on the discrete-time counterpart proposed in \cite{yan2022}.
In addition, when all agents share identical dynamics, i.e., $S_i(0) = S$ for all $i$, 
these conditions reduce exactly to those established in \cite{2} for homogeneous systems. 
Therefore, our synchronization criterion serves as a direct extension from homogeneous MASs in \cite{2} to heterogeneous MASs.
Furthermore, combining Theorem \ref{thm1} with Theorem \ref{thm2}, we provide an estimate of the synchronization rate (Theorem \ref{thm3}).
Finally, numerical simulations validate the results, 
showing that the theoretical convergence rate closely matches the actual rate.
\subsection*{3. Paper Organization}

The remainder of this paper is organized as follows:
Section 2 introduces basic graph theory concepts and formalizes the problems under study;
Section 3 presents the sufficient condition for asymptotic decoupling of system \eqref{eq:main-system} and provides its proof;
Section 4 applies the theoretical result to synchronization problems, proposing a control strategy and proving its effectiveness;
Section 5 provides numerical simulations to validate the theoretical findings;
Section 6 concludes the paper.
\section{Preliminaries and Problem Formulation}
\subsection{Notation}
We adopt the following conventions:
\begin{itemize}
    \item $\rho(A)$ denotes the spectral radius of a square matrix $A$, i.e., the maximum modulus of its eigenvalues.
    \item $\|X\|$ denotes the 2-norm of $X \in \mathbb{C}^n$
    \[
        \|X\| = \sqrt{X^H X}.
    \]
    \item The norm of a square matrix $A \in \mathbb{C}^{n \times n}$ is the induced operator norm
    \[
        \|A\| = \max_{\|X\|=1} \|AX\|.
    \]
\end{itemize}

\subsection{Problem 1: Multi-Agent Systems Synchronization Control}
\label{subsec:graph-theory}
We study a synchronization problem for MASs.
The communication topology is described by a connected, simple, undirected weighted graph $\mathcal{G} = (\mathcal{V}, \mathcal{E})$, where
\begin{itemize}
    \item $\mathcal{V} = \{1, 2, \ldots, N\}$ is the set of nodes, and each node represents an agent in the MASs;
    \item $\mathcal{E} \subseteq \mathcal{V} \times \mathcal{V}$ is the set of edges, representing connections between agents;
    \item if $\{i, j\} \in \mathcal{E}$, then agents $i$ and $j$ can communicate or interact.
\end{itemize}
The adjacency matrix $A = [a_{i,j}] \in \mathbb{R}^{N \times N}$ is defined as
\[
a_{i,j} = 
\begin{cases}
> 0, & \text{if } \{i, j\} \in \mathcal{E}, \\
0, & \text{otherwise},
\end{cases}
\quad \text{with } a_{i,i} = 0.
\]
The Laplacian matrix $L = [l_{i,j}] \in \mathbb{R}^{N \times N}$ is defined as
\[
l_{i,j} =
\begin{cases}
-a_{i,j}, & i \neq j, \\
\sum_{k \neq i} a_{i,k}, & i = j.
\end{cases}
\]
As $\mathcal{G}$ is connected, simple, and undirected weighted, the Laplacian matrix satisfies
\begin{itemize}
\item $L$ has eigenvalue $0$ with eigenvector $\mathbf{1}$ (the all-ones vector),
\item $L$ is symmetric positive semidefinite, and there exists an orthogonal matrix $Q$ such that
\begin{equation}\label{Q}
    L Q = Q \operatorname{diag}(\lambda_1, \lambda_2, \dots, \lambda_N),
\end{equation}
\end{itemize}
where $\lambda_i$ are the eigenvalues of $L$, ordered as
\[
0 = \lambda_1 < \lambda_2 \leq \lambda_3 \leq \cdots \leq \lambda_N.
\]

We consider the MASs based on the above adjacency matrix of the graph $\mathcal{G}$:
\begin{subequations}\label{eq:whole_system}
\begin{align}
\xi_i(t+1) &= S_i(t) \xi_i(t) + B u_i(t), \label{eq:system_a}\\
u_i(t) &= K_i(t) \sum_{j=1}^{N} a_{i,j} (\xi_j(t) - \xi_i(t)), \label{eq:system_b}\\
S_i(t+1) &= S_i(t) + c \sum_{j=1}^{N} a_{i,j} (S_j(t) - S_i(t)), \label{eq:system_c}
\end{align}
\end{subequations}
where
\begin{itemize}
    \item $\xi_i(t) \in \mathbb{R}^p$ is the state of the $i$-th agent,
    \item $B \in \mathbb{R}^p$ is the input matrix,
    \item $S_i(t) \in \mathbb{R}^{p \times p}$ is the system matrix of the $i$-th agent,
    \item $K_i(t) \in \mathbb{R}^p$ is the control matrix depending on $S_i(t)$,
    \item The parameter $c \in \mathbb{R}$ is the coupling strength.
\end{itemize}
\begin{definition}
    The multi-agent systems \eqref{eq:whole_system} is said to achieve synchronization if for any initial values $\xi_1(0), \xi_2(0), \ldots, \xi_N(0)$, 
    the corresponding solutions satisfy
    \begin{equation*}
        \lim_{t \to \infty} \|\xi_i(t) - \xi_j(t)\| = 0, \quad \forall 1\leq i ,j \leq N.
    \end{equation*}
\end{definition}
Our goal is to find suitable control matrices $K_i(t)$ and coupling strength $c$ such that the system \eqref{eq:whole_system} achieves synchronization.

To simplify \eqref{eq:whole_system}, introduce the following notation:
\begin{align*}
S(t) &= \operatorname{diag}(S_1(t), S_2(t), \ldots, S_N(t)),\\
\hat{S}(t) &= (S_1(t)^T, S_2(t)^T, \ldots, S_N(t)^T)^T,\\
H(t) &= \operatorname{diag}(B K_1(t), B K_2(t), \ldots, B K_N(t)),\\
\xi(t) &= (\xi_1(t)^T, \xi_2(t)^T, \ldots, \xi_N(t)^T)^T.
\end{align*}
Substituting \eqref{eq:system_b} into \eqref{eq:system_a}, system \eqref{eq:whole_system} can be rewritten as

\begin{subequations}\label{eq:newwhole_system}
\begin{align}
    \xi(t+1) &= [S(t) - H(t) (L \otimes I_p)] \xi(t), \label{eq:newsystem_ab} \\
    \hat{S}(t+1) &= ((I_N - c L) \otimes I_p) \hat{S}(t). \label{eq:newsystem_c}
\end{align}
\end{subequations}

\subsection{Problem 2: Asymptotic Decoupling of LTV Systems}
\label{subsec:problem-formulation}
Under a suitable coordinate transformation, 
system \eqref{eq:newsystem_ab} can be converted into the form of system \eqref{eq:main-system},
where $X_t \in \mathbb{C}^n$ is the system state
and $A$ is a block-diagonal matrix
\begin{equation}
\label{eq:A-block}
A = \begin{pmatrix}
A^* & 0 \\
0 & A^s
\end{pmatrix},
\end{equation}
$A^s$ characterizes the synchronization among agents, 
while $A^*$, which is exactly $S_{\infty}$, 
characterizes the collective motion of the agents.
For system \eqref{eq:main-system}, we assume that
\begin{itemize}
    \item $\rho(A^{*}) \geq 1$ (unstable modes),
    \item $\rho(A^{s}) < 1$ (stable modes),
    \item $P_t$ is a time-varying perturbation satisfying $\|P_t\| \to 0$.
\end{itemize}
Correspondingly, the state $X_t$ can be decomposed as
\[
X_t = \begin{pmatrix} Y_t \\ Z_t \end{pmatrix}.
\]
Write \( P_t \) in proper block form:
\[
P_t = \begin{pmatrix}
P_{1,t} & P_{2,t} \\
P_{3,t} & P_{4,t}
\end{pmatrix},
\]
so that system \eqref{eq:main-system} can be rewritten as
\begin{equation}
\label{eq:decomposed-system}
\begin{pmatrix} Y_{t+1} \\ Z_{t+1} \end{pmatrix} = 
\left[ \begin{pmatrix} A^* & O \\ O & A^s \end{pmatrix} + 
\begin{pmatrix} P_{1,t} & P_{2,t} \\ P_{3,t} & P_{4,t} \end{pmatrix} \right]
\begin{pmatrix} Y_t \\ Z_t \end{pmatrix}.
\end{equation}

\begin{definition}
    System \eqref{eq:decomposed-system} is said to achieve asymptotic decoupling, if $\|Z_t\| \rightarrow 0$ for any initial values.
\end{definition}

It is well known that when system \eqref{eq:decomposed-system} achieves asymptotic decoupling, 
MASs \eqref{eq:newsystem_ab} achieves synchronization.
And \cite{1} shows that in continuous-time case, 
system \eqref{eq:decomposed-system} achieves asymptotic decoupling if $P_{3,t}$ decays exponentially at a rate faster than $1/\rho(A^*)$. 
Therefore, we aim to establish the same result for asymptotic decoupling problem in the discrete-time case, 
from which we can derive the synchronization conditions for system \eqref{eq:newsystem_ab}.
\section{Sufficient Condition for Asymptotic Decoupling of LTV Systems}
\label{sec:sufficient-conditions}

In this section, we present a sufficient condition for asymptotic decoupling of system \eqref{eq:decomposed-system}: 
if $P_{3,t}$ decays exponentially at a rate faster than $1/\rho(A^{*})$, then $\|Z_t\|$ converges exponentially to zero,
which is summarized by the following theorem.

\begin{theorem}
\label{thm1}
    For system \eqref{eq:decomposed-system}, suppose that there exist $0 \leq \kappa <1/\rho(A^{*})$ and $M > 0$ 
    such that $\|P_{3,t}\| \leq M \kappa^t$.
    Then for any $r$ satisfying $\max \left( \rho(A^{s}), \kappa \rho(A^{*}) \right)<r <1 $,
    there exists $M_{r}$ such that
    \begin{align*}
    \|Z_t\| \leq M_{r}   \|X_0\| {r} ^t, \, \text{for any initial state} \, X_0.
    \end{align*}
\end{theorem}

To prove Theorem \ref{thm1}, we introduce three lemmas.
\begin{lemma}\label{lem1}\cite{horn2012matrix}
For any square matrix $A \in \mathbb{C}^{n\times n}$ and any $\epsilon > 0$, there exists an invertible matrix $Q$ such that
\begin{align}
    \|Q^{-1} A Q\| < \rho(A) + \epsilon.
\end{align}
\end{lemma}

\begin{lemma}\label{lem2}
Suppose that $P_t \to \mathbf{0}$. Then for any square matrix $A\in \mathbb{C}^{n\times n}$ and $\epsilon > 0$, there exists $M > 0$ such that
\[
\frac{\left\| \prod_{i=0}^{t-1} (A + P_i) \right\|}{(\rho(A) + \epsilon)^t} < M.
\]
\end{lemma}
\begin{proof}
By Lemma \ref{lem1}, there exists an invertible matrix $Q$ such that $\|Q^{-1} A Q\| < \rho(A) + \frac{\epsilon}{2}$.
We have
\begin{equation}\label{eq:lem2:1}
\begin{split}
\prod_{i=0}^{t-1} (A + P_i) &= \prod_{i=0}^{t-1} Q(Q^{-1} A Q + Q^{-1} P_i Q) Q^{-1}\\
&= Q \left[ \prod_{i=0}^{t-1} (Q^{-1} A Q + Q^{-1} P_i Q) \right] Q^{-1}.
\end{split}
\end{equation}
Taking the norm on both sides of \eqref{eq:lem2:1}, we obtain
\begin{equation}\label{eq:lem2:2}
\begin{split}
&\left\| \prod_{i=0}^{t-1} (A + P_i) \right\| \\
\leq &M'  \left\| \prod_{i=0}^{t-1} (Q^{-1} A Q + Q^{-1} P_i Q) \right\|  \\
\leq &M'\prod_{i=0}^{t-1}  \left[\|Q^{-1} A Q\|+ \|Q^{-1}\|  \|P_i\|  \|Q\| \right],
\end{split}
\end{equation}
where $M'=\|Q\|  \|Q^{-1}\|$.\\
Since $\|P_i\| \to 0$, there exists $N_1$ such that for any $i > N_1$,
$\|Q^{-1}\|  \|P_i\|  \|Q\| < \frac{\epsilon}{4}$.
Thus for $i > N_1$,we have
\begin{equation}\label{eq:lem2:3}
    \|Q^{-1} A Q\| + \|Q^{-1}\| \cdot \|P_i\| \cdot \|Q\| < \rho(A) + \frac{3}{4} \epsilon.
\end{equation}
It follows from \eqref{eq:lem2:2} and \eqref{eq:lem2:3} that for any $i> N_1$, we have
\begin{equation}\label{eq:lem2:4}
\left\| \prod_{i=0}^{t-1} (A + P_i) \right\| 
\leq M_1 \left( \rho(A) + \frac{3}{4} \epsilon \right)^{t-N_1-1}, 
\end{equation}
where
\[
    M_1=\|Q\|  \|Q^{-1}\|  \prod_{i=0}^{N_1} ( \|Q^{-1} A Q\|+ \|Q^{-1}\|\|P_i\| \|Q\| ).
\]
Dividing both sides of \eqref{eq:lem2:4} by $(\rho(A) + \epsilon)^t$, we obtain that
\begin{equation}\label{eq:lem2:5}
\frac{\left\| \prod_{i=0}^{t-1} (A + P_i) \right\|}{(\rho(A) + \epsilon)^t}
\leq 
M_2 \left( \frac{\rho(A) + \frac{3}{4} \epsilon}{\rho(A) + \epsilon} \right)^t, \quad \forall \, t > N_1+1,
\end{equation}
where
\[
    M_2=\frac{\|Q\|  \|Q^{-1}\|  \prod_{i=0}^{N_1} \left( \|Q^{-1} A Q\| + \|Q^{-1}\|  \|P_i\|  \|Q\| \right)}{(\rho(A) + \frac{3}{4} \epsilon)^{N_1+1}}.
\]
The right-hand side of \eqref{eq:lem2:5} tends to $0$, hence is bounded.
Thus the left-hand side is bounded. Since finitely many terms do not affect boundedness, there exists $M > 0$ such that
\[
\frac{\left\| \prod_{i=0}^{t-1} (A + P_i) \right\|}{(\rho(A) + \epsilon)^t} < M,\quad \forall \, t.
\]
\end{proof}
\begin{lemma}\label{lem3}
Let $\{a_t\}_{t=0}^{\infty}$ be a nonnegative sequence satisfying $a_{t+1} \leq b_t + \lambda_t a_t$, with $b_t, \lambda_t \geq 0$. 
Suppose that
\begin{itemize}
\item $\lambda_t \to \lambda < 1$;
\item There exist $M > 0$ and $\beta < 1$ such that $b_t \leq M \beta^t$.
\end{itemize}
Then for any $r$ with $\max(\beta, \lambda) < r < 1$, there exists $M'>0$ such that $a_t \leq M' (r)^t$.
\end{lemma}
\begin{proof}
As $\lambda_t \to \lambda$, for any $\epsilon >0$ with $\epsilon \neq \beta-\lambda$, there exists $T$ such that
\begin{equation*}
    \lambda_t < \lambda + \epsilon,\quad \forall \, t \geq T.
\end{equation*}
Therefore, we have
\begin{equation*} 
    a_{t+1} \leq M \beta^t + (\lambda + \epsilon) a_t,\quad \forall \, t \geq T.
\end{equation*}
Considering a new sequence $\{c_t\}_{t=T}^{\infty}$ satisfying 
\begin{subequations}
\begin{align}
    c_{T}&=a_{T},\\
    c_{t+1} &= M \beta^t + (\lambda + \epsilon) c_t,\quad \forall \, t \geq T, \label{eq:lem3:1}
\end{align}
\end{subequations}
by induction we have
\begin{equation}
    c_t \geq a_t ,\quad \forall \, t \geq T . \label{eq:lem3:2}
\end{equation} 
\eqref{eq:lem3:1} is a first-order linear nonhomogeneous recurrence relation, with solution
\[
c_t = \frac{M}{\beta - (\lambda + \epsilon)} \beta^t + C (\lambda + \epsilon)^t, \quad \forall \, t \geq T,
\]
where $C$ depends on $c_{T}$.\\
Thus there exists $M'$ such that
\begin{equation}\label{eq:lem3:c}
c_t \leq M' \max(\beta, \lambda + \epsilon)^t.
\end{equation}
It follows from \eqref{eq:lem3:2} and \eqref{eq:lem3:c} that
\begin{equation}\label{eq:lem3:3}
    a_t \leq M' \max(\beta, \lambda + \epsilon)^t.
\end{equation}

Since $\epsilon > 0$ with $\epsilon \neq \beta-\lambda$ can be chosen arbitrarily small, 
the quantity $\max(\beta, \lambda + \epsilon)$ can be made arbitrarily close to $\max(\beta, \lambda)$. 
Consequently, for any $r$ satisfying $\max(\beta, \lambda) < r < 1$, 
we may select $\epsilon$ sufficiently small such that $\max(\beta, \lambda + \epsilon) < r$. 
Then, from \eqref{eq:lem3:3}, there exists $M' > 0$ (depending on the chosen $\epsilon$) such that
\[
a_t \leq M' r^t.
\]
This completes the proof.
\end{proof}

With these lemmas, we now prove Theorem \ref{thm1}.

\begin{proof}
Let $\{X_t\}_{t=0}^{\infty}$ be the trajectory of system \eqref{eq:decomposed-system} generated from an arbitrary initial condition $X_0$.
Apply Lemma \ref{lem1} to $A^s$: for any $0<\delta< 1-\rho(A^{s})$, there exists an invertible matrix $R$ such that $\|R^{-1} A^s R\| < \rho(A^{s}) + \delta$.

Consider the following change of variables:
\begin{subequations}
\begin{align}
    Y = Y', \label{eq:thm1:1} \\
    Z = R Z'. \label{eq:thm1:2}
\end{align}
\end{subequations}
From \eqref{eq:thm1:1}\eqref{eq:thm1:2}, system \eqref{eq:decomposed-system} becomes
\begin{subequations}
\begin{align}
Y_{t+1}' &= \bigl( A^* + P_{1,t} \bigr) Y_t' + \bigl( P_{2,t} R \bigr) Z_t', \\
Z_{t+1}' &= \bigl( R^{-1} P_{3,t} \bigr) Y_t' + \bigl( R^{-1} A^s R + R^{-1} P_{4,t} R \bigr) Z_t'. \label{eqZ_k}
\end{align}
\end{subequations}
In particular, taking the norm on both sides of \eqref{eqZ_k}, we obtain
\begin{align}\label{ineqZ_k}
\|Z'_{t+1}\| \leq \|R^{-1} P_{3,t} Y'_t\| + \|R^{-1} A^s R + R^{-1} P_{4,t} R\|  \|Z'_t\|.
\end{align}

In the following we will verify that \eqref{ineqZ_k} satisfies the conditions of Lemma \ref{lem3}.

First, we estimate $\|R^{-1} P_{3,t} Y'_t\|$.
\begin{equation}\label{eq:thm1:3}
\begin{split}
\|Y'_t\| &= \|Y_t\| \leq \|X_t\| \leq \biggl\| \prod_{i=0}^{t-1} (A + P_i) X_0 \biggr\| \\
&\leq \biggl\| \prod_{i=0}^{t-1} (A + P_i) \biggr\| \|X_0\|.
\end{split}
\end{equation}
By Lemma \ref{lem2}, for any $0<\epsilon<\kappa^{-1}-\rho(A^{*})$, there exists $M'$ such that
\begin{equation}\label{eq:thm1:4}
\frac{\left\| \prod_{i=0}^{t-1} (A + P_i) \right\|}{(\rho(A^{*}) + \epsilon)^t} < M'.
\end{equation}
It follows from \eqref{eq:thm1:3} and \eqref{eq:thm1:4} that
\begin{equation}\label{eq:thm1:5}
    \|Y'_t\| \leq M'\|X_0\| (\rho(A^{*}) + \epsilon)^t.
\end{equation}
Combining inequality \eqref{eq:thm1:5} and hypothesis $\|P_{3,t}\| \leq M \kappa^t$, we have
\begin{equation}\label{eq:thm1:8}
\begin{split}
\|R^{-1} P_{3,t} Y'_t\| &\leq \|R^{-1}\| \|P_{3,t}\| \|Y'_t\| \\
&\leq M M' \|R^{-1}\| \|X_0\| ((\rho(A^{*}) + \epsilon)\kappa) ^t,
\end{split}
\end{equation}
where $(\rho(A^{*}) + \epsilon) \kappa<1$.

Second, since $\|P_t\| \to 0$ and by the continuity of the operator norm, we have
\begin{equation}\label{eq:thm1:7}
\|R^{-1} A^s R + R^{-1} P_{4,t} R\| \to \|R^{-1} A^s R\| <\rho(A^{s}) + \delta < 1.
\end{equation}

From \eqref{eq:thm1:8}-\eqref{eq:thm1:7} we can see that \eqref{ineqZ_k} satisfies the conditions of Lemma \ref{lem3}.
Applying Lemma \ref{lem3} to \eqref{ineqZ_k}, we conclude that for any $r$ satisfying
\[
\max\left( (\rho(A^{*})+\epsilon)\kappa, \rho(A^{s}) + \delta \right) < r < 1,
\]
there exists $M_1$ such that
\[
\|Z'_t\| \leq M_1 r^t.
\]

Since $Z$ and $Z'$ are related by the invertible linear transformation $R$, a similar estimate holds for $Z_t$: 
for any $r$ satisfying
\[
\max\left( (\rho(A^{*})+\epsilon)\kappa, \rho(A^{s}) + \delta \right) < r < 1,
\]
there exists $M_2$ such that
\[
\|Z_t\| \leq M_2 r^t.
\]

By the arbitrariness of $\delta$ and $\epsilon$, for any $r$ satisfying
\[
    \max\left(\rho(A^{*})\kappa, \rho(A^{s}) \right) < r < 1,
\]
there exists $M_3$ such that
\[
\|Z_t\| \leq M_3 r^t.
\]

Finally, because system \eqref{eq:decomposed-system} is finite-dimensional and linear, for any $r$ satisfying
$\max\left(  \rho(A^{*})\kappa, \rho(A^{s}) \right) < r < 1$,
there exists $M_{r}$ such that 
\[
\|Z_t\| < M_{r} \|X_0\| r ^t, \quad \forall \,  X_0.
\]
This completes the proof.
\end{proof}
\begin{remark}
    In \cite{1}, which deals with continuous-time systems, the perturbation $P_t$ is required to decay exponentially; 
    in particular, the decay rate of $P_{3,t}$ must be smaller than $\rho(A^*)^{-1}$. 
    However, as shown above, for a discrete-time LTV system of the form \eqref{eq:decomposed-system}, 
    asymptotic decoupling does not require the other components of $P_t$ to decay exponentially; 
    it suffices that $P_{3,t}$ decays with rate smaller than $\rho(A^*)^{-1}$ and the rest tend to zero. 
    By a similar line of reasoning, a similar conclusion holds for continuous-time systems: 
    only the exponential decay of $P_{3,t}$ with a sufficiently small rate is needed,
    while the other blocks need only tend to zero.
\end{remark}

\section{Application to System Synchronization}
\label{sec:application-synchronization}
In this section, we select a suitable parameter $c$ and design an update rule $K_i(t)$ that depends only on $S_i(t)$ 
to make system \eqref{eq:whole_system} achieve synchronization. 

To prove the effectiveness of this control design, we transform the multiagent systems \eqref{eq:newsystem_ab} into the form \eqref{eq:decomposed-system} 
and then apply Theorem \ref{thm1} to establish asymptotic decoupling for system \eqref{eq:decomposed-system},
which in turn implies synchronization for system \eqref{eq:whole_system}.

For system \eqref{eq:whole_system}, let $S_{\infty} = \frac{1}{N} \sum_{j=1}^{N} S_j(0)$,
we make the following assumptions:
\begin{assumption}\label{asm1}
$(S_{\infty}, B)$ is stabilizable.
\end{assumption}

\begin{assumption}\label{asm2}
$\rho(S_{\infty}) \geq 1$.
\end{assumption}

\begin{assumption}\label{asm3}
$\left( \prod_j |\lambda_j^u(S_{\infty})| \right)^{-1} > \frac{\lambda_N - \lambda_2}{\lambda_N + \lambda_2}$,
where $\lambda_j^u(S_{\infty})$ are the eigenvalues of $S_{\infty}$ with $|\lambda_j^u| \ge 1$.
\end{assumption}

The following lemmas are needed for the main result of this section.
\begin{lemma}\label{lem4}
For \eqref{eq:system_c}, under Assumptions \ref{asm2} and \ref{asm3}, 
there exists $M>0$ such that
\[
\|S_i(t)-S_{\infty}\| < M (\max(|1-c\lambda_2|,|1-c\lambda_N|))^t,\quad i = 1,\ldots,N,
\]
for any c satisfying
\[
    c \in (\frac{1-1/\rho(S_{\infty})}{\lambda_2},\frac{1+1/\rho(S_{\infty})}{\lambda_N}),
\]
where $M$ depends on $S_1(0), \ldots, S_N(0)$ and $\max(|1-c\lambda_2|,|1-c\lambda_N|)< 1/ \rho(S_{\infty})$.\\
In particular, when $c = \frac{2}{\lambda_N + \lambda_2}$,
\[
\|S_i(t) - S_{\infty}\| \leq M \left( \frac{\lambda_N - \lambda_2}{\lambda_N + \lambda_2} \right)^t,\quad i = 1,\ldots,N,
\]
this choice yields the fastest convergence rate.
\end{lemma}
\begin{proof}
Considering the function
\[
    f(x)=\max(|1-x\lambda_2|,|1-x\lambda_N|),
\]
we have $f(x)$ is strictly decreasing on $(-\infty, \frac{2}{\lambda_N+\lambda_2})$ 
and strictly increasing on $(\frac{2}{\lambda_N+\lambda_2}, +\infty)$.
Therefore, $f(x)$ attains its minimum 
\begin{equation}\label{eq:lem4:6}
    f(\frac{2}{\lambda_N+\lambda_2})=\frac{\lambda_N-\lambda_2}{\lambda_N+\lambda_2}.
\end{equation}
From Assumptions \ref{asm2} and \ref{asm3} we have
\begin{equation}\label{eq:lem4:7}
    \frac{\lambda_N - \lambda_2}{\lambda_N + \lambda_2}< \left( \prod_j |\lambda_j^u(S_{\infty})| \right)^{-1} \leq 1/\rho(S_{\infty}).
\end{equation}
It follows from \eqref{eq:lem4:6} and \eqref{eq:lem4:7} that there exists at least a solution to inequality
\begin{equation}\label{eq:lem4:8}
    f(x) < 1/ \rho(S_{\infty}).
\end{equation}
By solving inequality \eqref{eq:lem4:8}
we have
\begin{equation}\label{eq:lem4:9}
    x \in (\frac{1-1/\rho(S_{\infty})}{\lambda_2},\frac{1+1/\rho(S_{\infty})}{\lambda_N}).
\end{equation}
In the following, we take 
\[
    c \in (\frac{1-1/\rho(S_{\infty})}{\lambda_2},\frac{1+1/\rho(S_{\infty})}{\lambda_N}).
\]
From \eqref{eq:lem4:9} we have
\begin{equation}\label{eq:lem4:10}
    f(c) < 1/ \rho(S_{\infty}).
\end{equation}
Since $\lambda_2 \leq \cdots \leq \lambda_N$ and by \eqref{eq:lem4:10}, we have 
\begin{equation}\label{eq:lem4:3}
|1 - c\lambda_i| \leq f(c)< 1/ \rho(S_{\infty}),\quad  i=2 \ldots N.
\end{equation}
From \eqref{eq:newsystem_c}, we have
\begin{equation}\label{eq:lem4:1}
\hat{S}(t+1) = ((I_N - c L)^t \otimes I_p) \hat{S}(0).
\end{equation}
It follows from \eqref{Q} that
\begin{equation}\label{eq:lem4:4}
    I_N - c L = Q \operatorname{diag}(1, 1 - c\lambda_2, \ldots, 1 - c\lambda_N) Q^{-1},
\end{equation}
where the first column of $Q$ is $\frac{1}{\sqrt{N}} \mathbf{1}$.\\
Combining \eqref{eq:lem4:3} and \eqref{eq:lem4:4}, we have
\begin{equation}\label{eq:lem4:5}
(I_N - c L)^t \to
\begin{bmatrix}
\frac{1}{N} \mathbf{1} & \frac{1}{N} \mathbf{1} & \cdots & \frac{1}{N} \mathbf{1}
\end{bmatrix},
\end{equation}
It follows from \eqref{eq:lem4:1} and \eqref{eq:lem4:5} that
\begin{equation}
S_i(t) \to \frac{1}{N} \sum_{j=1}^{N} S_j(0) = S_{\infty}, \quad  i = 1, \dots, N,
\end{equation}
with the convergence rate $f(c)$.
\end{proof}

\begin{remark}
To show that \(\frac{\lambda_N - \lambda_2}{\lambda_N + \lambda_2}\) is the fastest convergence rate, it suffices to note that
$\min_{x} \bigl\{\max (|1 - x\lambda_2|, |1 - x\lambda_N|) \bigr\}=\frac{\lambda_N - \lambda_2}{\lambda_N + \lambda_2}$. 
(see, e.g., \cite{LinXiao2003} for a detailed analysis of optimal linear iterations for distributed averaging).
\end{remark}

\begin{lemma}\label{lem5}\cite{HENGSTERMOVRIC2013414}
Under Assumptions \ref{asm1} and \ref{asm2}, for any $\eta < \bigl(\prod_j |\lambda_j^u(S_{\infty})|\bigr)^{-1}$, 
there exists a positive definite matrix $P$ such that the following Riccati inequality holds:
\begin{align}\label{eq:Riccati}
P - S_{\infty}^T P S_{\infty} + (1 - \eta^2) \frac{S_{\infty}^T P B B^T P S_{\infty}}{B^T P B} > 0.
\end{align}
\end{lemma}

We now state the main theorem of this section.
\begin{theorem}\label{thm2}
For system \eqref{eq:whole_system}, under Assumptions~\ref{asm1}--\ref{asm3}, 
there exist suitable $c$ and $K_i(t)$ that depends only on $S_i(t)$ such that the multiagent systems synchronizes.
Specifically, we take
\[
c \in (\frac{1-1/\rho(S_{\infty})}{\lambda_2},\frac{1+1/\rho(S_{\infty})}{\lambda_N}),
\]
\[
K_i(t) = \frac{2}{\lambda_2 + \lambda_N} \frac{B^T P S_i(t)}{B^T P B},
\]
where $P$ is a solution of the Riccati inequality \eqref{eq:Riccati}.
\end{theorem}

\begin{proof}
By Assumption \ref{asm3} and Lemma \ref{lem5}, we can choose $\eta$ satisfying
\[
\frac{\lambda_N - \lambda_2}{\lambda_N + \lambda_2} \leq \eta < (\prod_j |\lambda_j^u(S_{\infty})|)^{-1}.
\]
Then the result $P$ for \eqref{eq:Riccati} satisfies
\begin{equation}\label{eq1}
P - S_{\infty}^T P S_{\infty} + \frac{4 \lambda_2 \lambda_N}{(\lambda_2 + \lambda_N)^2} \frac{S_{\infty}^T P B B^T P S_{\infty}}{B^T P B} > 0.
\end{equation}
Considering the control:
\[
c \in (\frac{1-1/\rho(S_{\infty})}{\lambda_2},\frac{1+1/\rho(S_{\infty})}{\lambda_N}),
\]
\[
K_i(t) = \frac{2}{\lambda_2 + \lambda_N} \frac{B^T P S_i(t)}{B^T P B},
\]

we now prove the effectiveness of this control.
Let
\begin{equation}\label{eq:K}
K_{\infty} = \frac{2}{\lambda_2 + \lambda_N} \frac{B^T P S_{\infty}}{B^T P B}.
\end{equation}
From Lemma \ref{lem4}, there exists $M'>0$ such that
\begin{equation}\label{eq2}
\|K_i(t) - K_{\infty}\| < M' \kappa^t,
\end{equation}
where $\kappa =\max(|1-c\lambda_2|,|1-c\lambda_N|)<1/\rho(S_{\infty})$.\\
Rewrite \eqref{eq:newsystem_ab} as
\begin{equation}\label{eq:thm2:1}
\xi(t+1) = [(I_N \otimes S_{\infty}) - (L \otimes (B K_{\infty}))] \xi(t) + R(t) \xi(t),
\end{equation}
\begin{equation*}
R(t) = S(t) - (I_N \otimes S_{\infty}) - \left[H(t) (L \otimes I_p) - L \otimes (B K_{\infty})\right]. \label{eq5}
\end{equation*}
It follows from Lemma \ref{lem4} and \eqref{eq2} that there exists $M''$ such that
\begin{align}\label{eq_R}
\|R(t)\| < M'' \kappa^t,
\end{align}
where $\kappa <1/\rho(S_{\infty})$.

Consider the following transformation:
\begin{align}\label{eq3}
T &= \begin{bmatrix}
\frac{1}{\sqrt{N}} \mathbf{1} & U
\end{bmatrix} \otimes I_p,\\
\xi &= T \begin{bmatrix}
\sigma \\ \zeta
\end{bmatrix}, \quad \sigma \in \mathbb{R}^p,\; \zeta \in \mathbb{R}^{(N-1)p}.
\end{align}
Then the transformed system \eqref{eq:thm2:1} becomes
\begin{equation*}
\begin{split}
\begin{bmatrix}
\sigma(t+1) \\ \zeta(t+1)
\end{bmatrix}
=
&T^{-1}(I_N \otimes S_{\infty} - L \otimes B K_{\infty}) T
\begin{bmatrix}
\sigma(t) \\ \zeta(t)
\end{bmatrix}\\
&+ T^{-1} R(t) T \begin{bmatrix}
\sigma(t) \\ \zeta(t)    
\end{bmatrix},
\end{split}
\end{equation*}
i.e.,
\begin{equation}\label{eq:Lambda}
\begin{bmatrix}
\sigma(t+1) \\ \zeta(t+1)
\end{bmatrix}
= \Lambda
\begin{bmatrix}
\sigma(t) \\ \zeta(t)
\end{bmatrix}
+ T^{-1} R(t) T \begin{bmatrix}
\sigma(t) \\ \zeta(t)
\end{bmatrix},
\end{equation}
where
\[
\Lambda = \operatorname{diag}(S_{\infty}, S_{\infty} - \lambda_2 B K_{\infty}, \ldots, S_{\infty} - \lambda_N B K_{\infty}).
\]
To apply Theorem \ref{thm1} to system \eqref{eq:Lambda}, we need to verify that
\[
\rho(S_{\infty} - \lambda_j B K_{\infty}) < 1  \qquad j = 2, \dots, N .
\]
For $S_{\infty} - \lambda_j B K_{\infty}$, we have
\begin{equation}\label{eq:thm2:2}
\begin{split}
&(S_{\infty} - \lambda_j B K_{\infty})^T P (S_{\infty} - \lambda_j B K_{\infty}) - P\\
&= S_{\infty}^T P S_{\infty} - P + \frac{4\lambda_j (\lambda_j - (\lambda_2 + \lambda_N))}{(\lambda_N + \lambda_2)^2} \frac{S_{\infty}^T P B B^T P S_{\infty}}{B^T P B}\\
&\leq S_{\infty}^T P S_{\infty} - P - \frac{4\lambda_2 \lambda_N}{(\lambda_N + \lambda_2)^2} \frac{S_{\infty}^T P B B^T P S_{\infty}}{B^T P B}.
\end{split}
\end{equation}
By \eqref{eq1}, the right-hand side of \eqref{eq:thm2:2} is negative definite, hence,
\begin{equation}\label{eq:thm2:3}
    \rho(S_{\infty} - \lambda_j B K_{\infty}) < 1 \qquad j = 2, \dots, N.
\end{equation}

By \eqref{eq_R} we have $T^{-1}R(t)T$ decays exponentially to $\mathbf{0}$ with the rate $\kappa$, and $\kappa <1/\rho(S_{\infty})$.
It follows from \eqref{eq:thm2:3} that system \eqref{eq:Lambda} satisfies the assumptions and the condition of Theorem \ref{thm1}, 
and hence is asymptotically decoupled, i.e., $\zeta(t) \to 0$.

Finally, from the transformation \eqref{eq3}, we have
\begin{equation}\label{eq:thm2:4}
\xi(t) - \frac{1}{\sqrt{N}} \mathbf{1} \otimes \sigma(t) = (U \otimes I_p) \zeta(t).
\end{equation}
It follows from \eqref{eq:thm2:4} and the fact $\zeta(t) \to \mathbf{0}$ that
\[
\|\xi_i(t) - \sigma(t)\| \to 0, \quad i=1 \ldots N.
\]
Hence,
\[
\|\xi_i(t) - \xi_j(t)\| \leq \|\xi_i(t) - \sigma(t)\| + \|\xi_j(t) - \sigma(t)\| \to 0, \quad \forall \, i, j.
\]
Thus, for any initial conditions $\xi_1(0),\ldots,\xi_N(0)$, system \eqref{eq:whole_system} achieves synchronization.
\end{proof}

\begin{remark}
Assumption \ref{asm2} is not strictly necessary; 
if it does not hold, one can simply take $c$ satisfying $0<c< 2/\lambda_N$ and $K_i(t) = \mathbf{0}$ such that system \eqref{eq:whole_system} achieves synchronization.
\end{remark}
\begin{remark}
    Theorem \ref{thm2} extends the result for homogeneous multiagent systems in \cite{2} to heterogeneous systems. 
    In \cite{2}, synchronization requires the common dynamics $S$ to satisfy
    \[
    \left(\prod_{j}|\lambda_{j}^{u}(S)|\right)^{-1} > \frac{\lambda_{N} - \lambda_{2}}{\lambda_{N} + \lambda_{2}}.
    \]
Theorem 2 replaces $S$ with the average dynamics $S_{\infty} = \frac{1}{N}\sum_{j=1}^{N}S_{j}(0)$. 
Thus, the conditions in Theorem 2 reduce to those in \cite{2} when all agents have identical dynamics.
\end{remark}
\begin{remark}
The condition in \cite{5} requires the agents' initial dynamics to be sufficiently close to each other. 
In contrast, Theorem \ref{thm2} remains valid even when the agents' initial dynamics exhibit arbitrarily large heterogeneity. 
See Example \ref{example1}.
\end{remark}
\begin{remark}
    To achieve synchronization, \cite{5} assumes $S_i(t) \to S_\infty$ at a rate $\kappa < 1/\|S_\infty\|$. 
    We, however, relax this condition to $\kappa < 1/\rho(S_\infty)$ 
    and, rather than simply assuming it, provide a design that enforces this rate.
\end{remark}
\begin{remark}
    Theorem \ref{thm2} guarantees the feasibility of autonomous synchronization for general heterogeneous multiagent systems 
    and constitute significant progress toward the discrete-time conjecture in \cite{yan2022}.
    However, the selection of the coupling strength $c$ and the positive definite matrix $P$ still relies on global information. 
    Therefore, the design of a fully distributed control strategy is left for future research.
\end{remark}
Combining Theorems \ref{thm1} and \ref{thm2}, we can further obtain the following estimate of the convergence rate for synchronization .

\begin{theorem}\label{thm3}
Under Assumptions \ref{asm1}, \ref{asm2}, \ref{asm3}, with the control design given in Theorem \ref{thm2}:
\[
c \in (\frac{1-1/\rho(S_{\infty})}{\lambda_2},\frac{1+1/\rho(S_{\infty})}{\lambda_N}),
\]
\[
K_i(t) = \frac{2}{\lambda_2 + \lambda_N} \frac{B^T P S_i(t)}{B^T P B},
\]
let
\[
r^* = \max\left( \rho(S_{\infty})f(c) , \max_{2 \leq j \leq N} (\rho(S_{\infty} - \lambda_j B K_{\infty}))\right),
\]
where $f(c)=\max(|1-c\lambda_2|,|1-c\lambda_N|)$ and $K_{\infty}$ is defined in \eqref{eq:K}.
Then for any $r$ satisfying $r^* < r < 1$, there exists $M_r$ such that for any initial state $\xi(0) \in \mathbb{R}^{Np}$,
\[
\|\xi_i(t) - \xi_j(t)\| < M_r \|\xi(0)\| r^t, \quad 1 \leq i, j \leq N.
\]
In particular, when $c = \frac{2}{\lambda_2 + \lambda_N}$, 
the expression for $r^*$ simplifies to
\[
    \max\left( \rho(S_{\infty}) \frac{\lambda_N - \lambda_2}{\lambda_N + \lambda_2}, \max_{2 \leq j \leq N} \rho(S_{\infty} - \lambda_j B K_{\infty}) \right).
\]
\end{theorem}

\begin{remark}
We only guarantee that for any $r$ with $r^* < r < 1$, the convergence is faster than $r^t$ (i.e., $\|\xi_i(t) - \xi_j(t)\| \leq C r^t$); 
whether this range $(r^{*},1)$ is optimal remains unknown.
\end{remark}
\begin{remark}
The choice $c = \frac{2}{\lambda_2 + \lambda_N}$ minimizes $r^*$ among all admissible $c$ in Theorem \ref{thm3},
and therefore yields the optimal convergence rate estimate within our analysis framework. 
Nevertheless, it is not claimed that this choice achieves the fastest synchronization rate for system \eqref{eq:whole_system}.
\end{remark}
\section{Numerical Simulation}
\label{sec:numerical-simulation}

To validate the effectiveness of Theorems \ref{thm2} and \ref{thm3}, am example with its numerical simulations is presented.

\begin{example}\label{example1}

Consider a weighted undirected network with $4$ agents, whose Laplacian matrix is
\[
L = \begin{bmatrix}
7 & -1 & -2 & -4 \\
-1 & 3 & -2 & 0 \\
-2 & -2 & 7 & -3 \\
-4 & 0 & -3 & 7
\end{bmatrix}.
\]
The eigenvalues are
\[
\lambda_1 = 0,\; \lambda_2 = 3.5714,\; \lambda_3 = 9.0784,\; \lambda_4 = 11.3503.
\]
The state dimension is $p=3$, and the input matrix is $B = [0, 0, 1]^T$.
The initial dynamics matrices are
\begin{align*}
S_1(0) &= \begin{bmatrix}0 & 0 & 0 \\ 0 & 0 & 0 \\ 0 & 0 & 0\end{bmatrix}, &
S_2(0) &= \begin{bmatrix}0 & 2 & 0 \\ 0 & 0 & 2 \\ 0 & 0 & 3\end{bmatrix}, \\
S_3(0) &= \begin{bmatrix}0 & 1 & 0 \\ 0 & 0 & 1 \\ 0 & 0 & 2\end{bmatrix}, &
S_4(0) &= \begin{bmatrix}0 & 1 & 0 \\ 0 & 0 & 1 \\ 0 & 0 & 1\end{bmatrix},
\end{align*}
which satisfy Assumption \ref{asm1} and \ref{asm2}.

We compute
\[
\left( \prod_j |\lambda_j^u(S_{\infty})| \right)^{-1} = 0.667 > 0.5214 = \left( \frac{\lambda_N - \lambda_2}{\lambda_N + \lambda_2} \right),
\]
therefore, Assumption \ref{asm3} is satisfied.

According to Theorem \ref{thm2}, we select $c=\frac{2}{\lambda_4 + \lambda_2}=0.1340$ and choose $\eta= 0.6$ such that $\eta$ satisfies 
\[
    0.5214=\frac{\lambda_N - \lambda_2}{\lambda_N + \lambda_2} \leq \eta < (\prod_j |\lambda_j^u(S_{\infty})|)^{-1}=0.667.
\]
By numerically solving the Riccati inequality \eqref{eq:Riccati}, we obtain the Lyapunov matrix $P$ as
\[
P = \begin{bmatrix}
0.10 & 0 & 0 \\
0 & 0.20 & 0 \\
0 & 0 & 2.3053
\end{bmatrix}.
\]
Therefore, we design control matrix
\[
    K_i(t) = \frac{2}{\lambda_2 + \lambda_4} \frac{B^T P S_i(t)}{B^T P B}.
\]
\begin{figure}[htbp]  
\centering
\subfloat[First state component]{
    \includegraphics[trim=0 150 0 150, clip, width=1.0\linewidth, height=5cm]{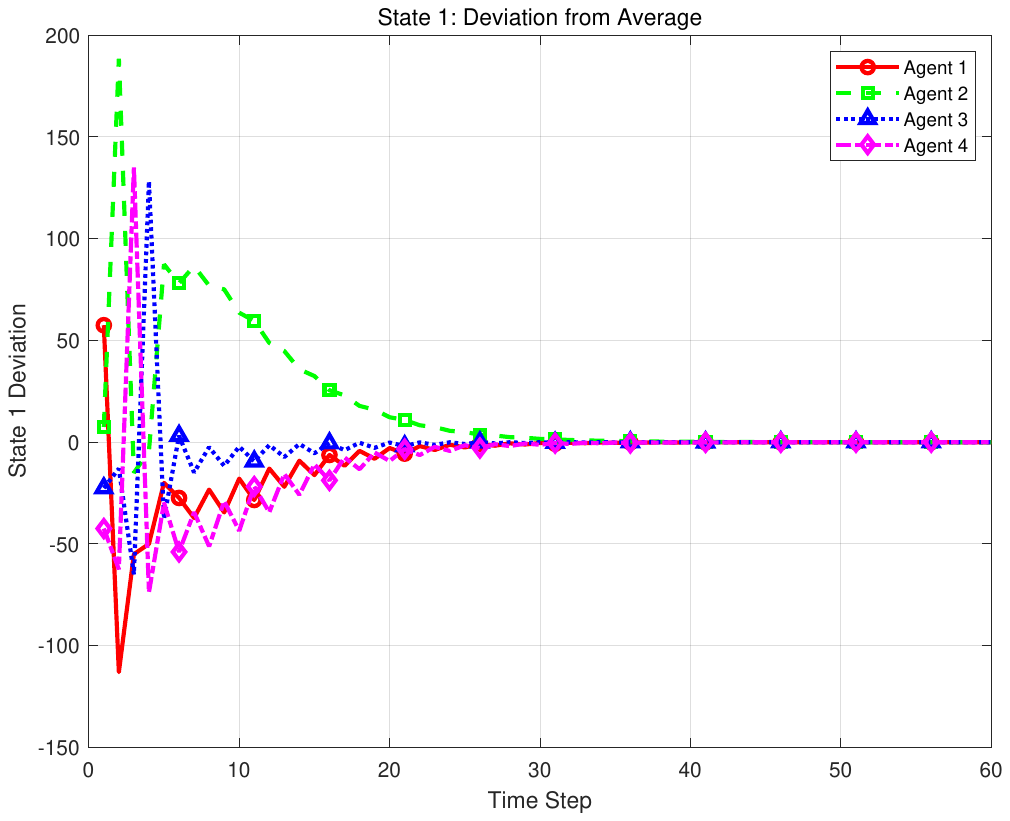} 
    \label{fig:state1}}

\subfloat[Second state component]{
    \includegraphics[trim=0 150 0 150, clip, width=1.0\linewidth, height=5cm]{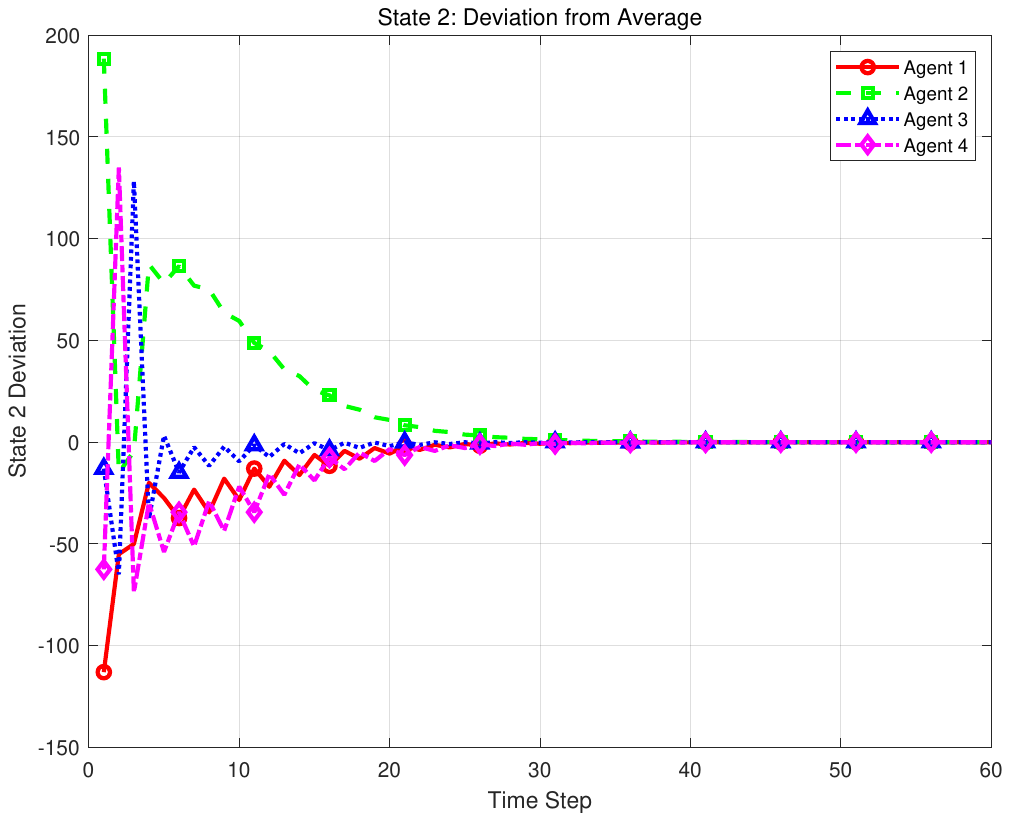}
    \label{fig:state2}}

\subfloat[Third state component]{
    \includegraphics[trim=0 150 0 150, clip, width=1.0\linewidth, height=5cm]{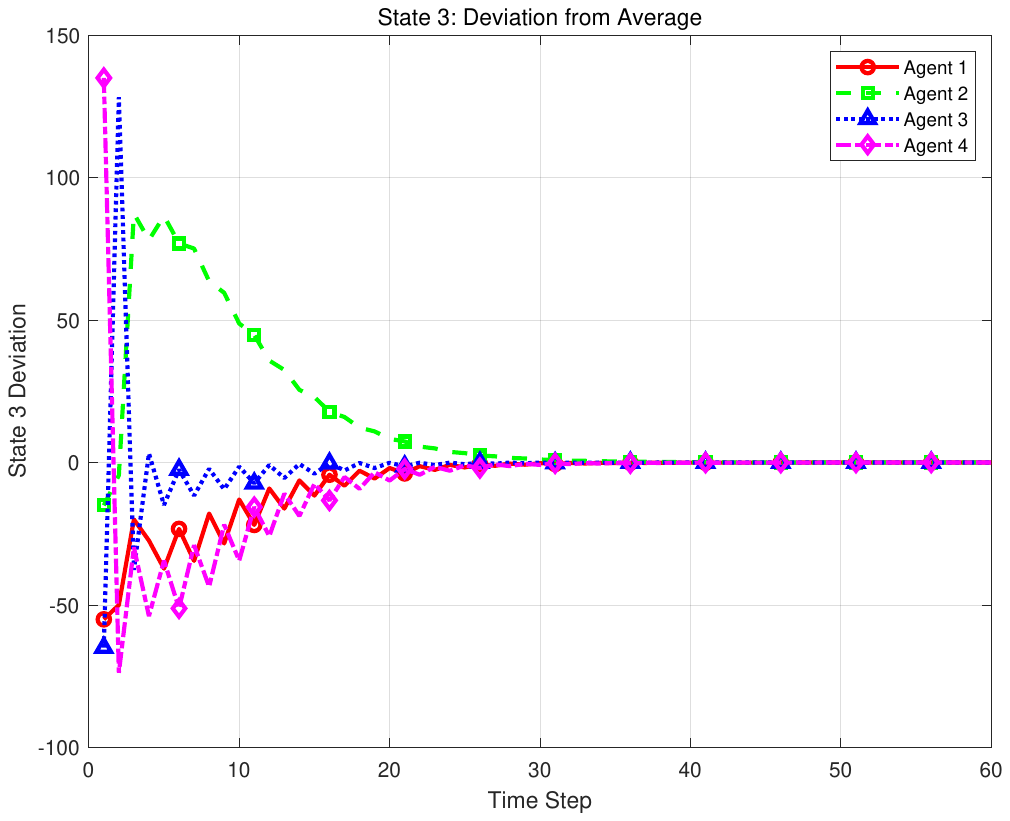}
    \label{fig:state3}}

\caption{Deviation of state components from the average.}
\label{fig:state_deviations}
\end{figure}

\begin{figure}[htbp]
\centering
\subfloat[Norm of state deviation (semilog)]{
    \includegraphics[trim=0 150 0 150, clip, width=1.0\linewidth, height=5cm]{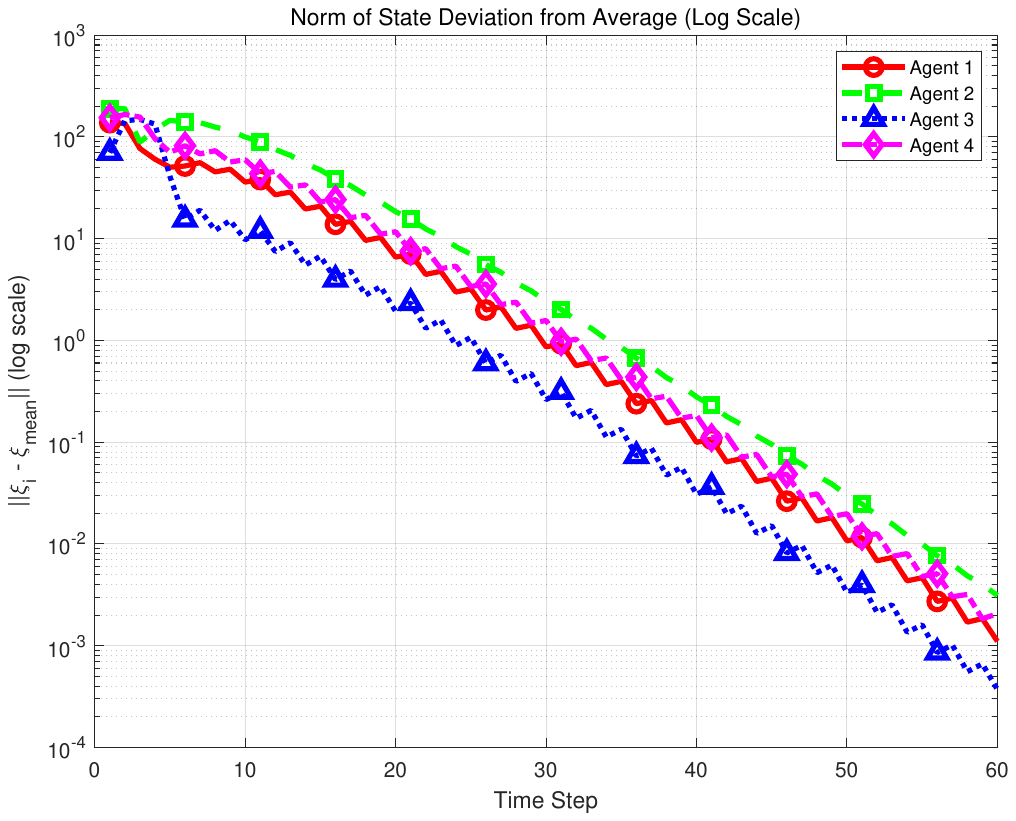}
    \label{fig4}}

\subfloat[Ratio to $0.9^t$ (semilog)]{
    \includegraphics[trim=0 150 0 150, clip, width=1.0\linewidth, height=5cm]{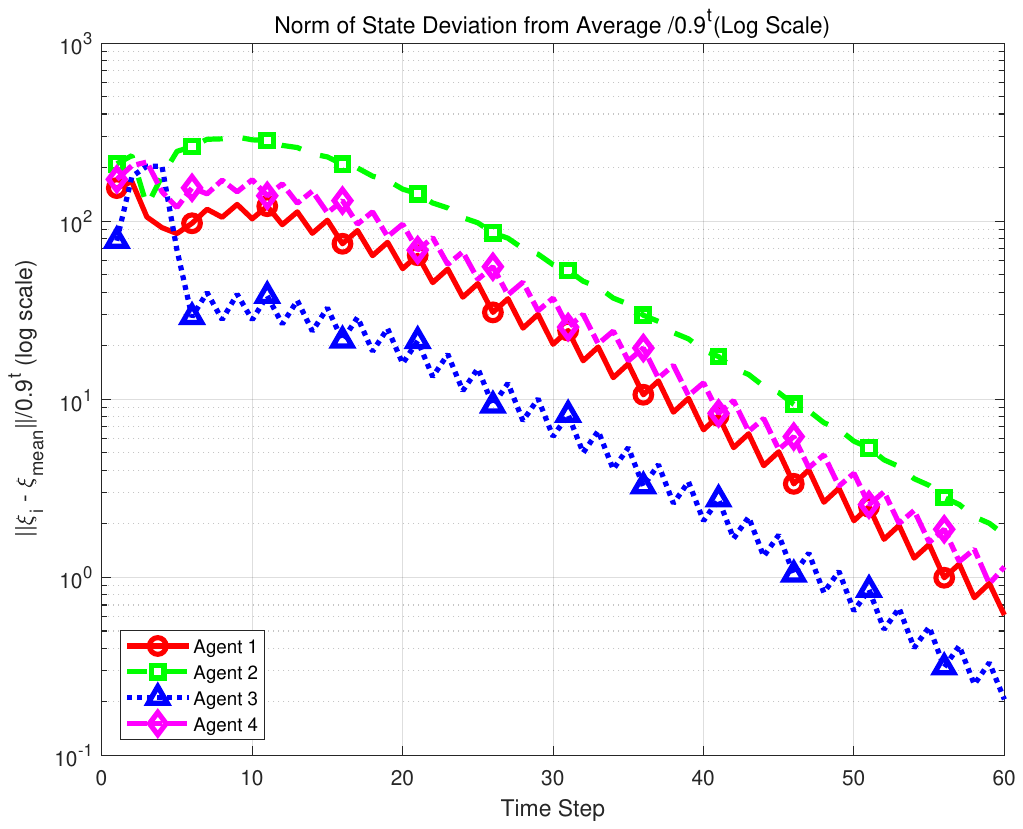}
    \label{fig5}}

\subfloat[Ratio to $0.8^t$ (semilog)]{
    \includegraphics[trim=0 150 0 150, clip, width=1.0\linewidth, height=5cm]{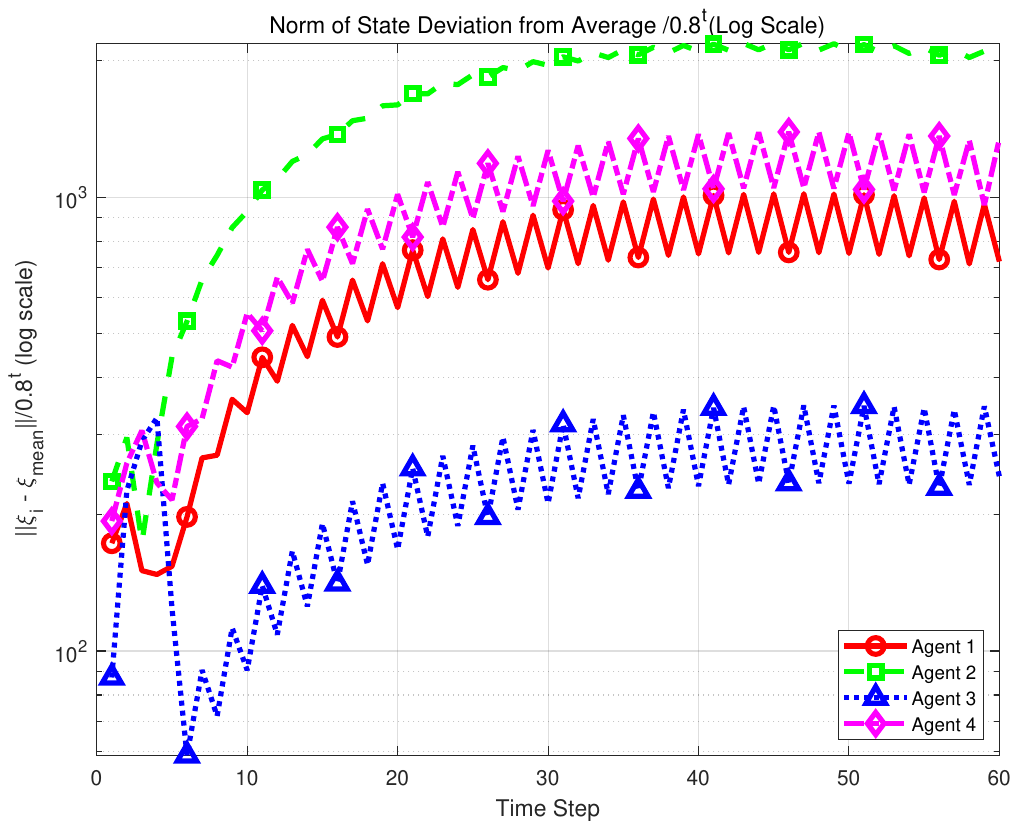}
    \label{fig6}}

\subfloat[Ratio to $0.7^t$ (semilog)]{
    \includegraphics[trim=0 150 0 150, clip, width=1.0\linewidth, height=5cm]{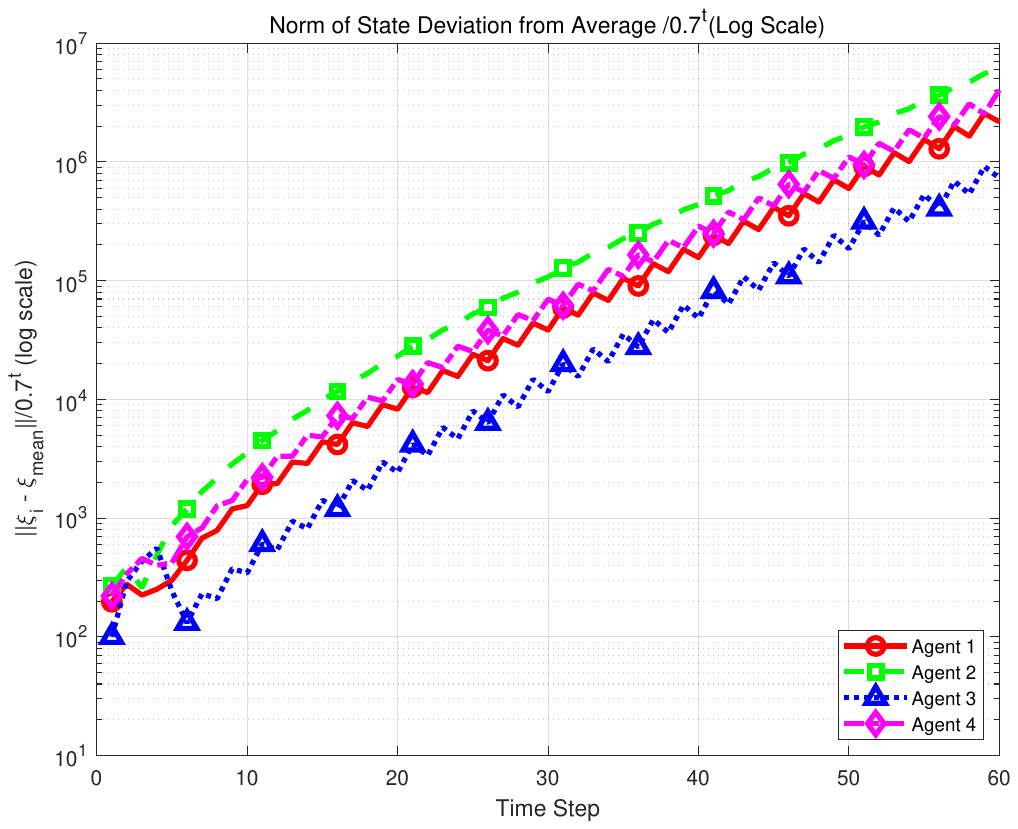}
    \label{fig7}}

\caption{State norm deviation and comparison with exponential decay.}
\label{fig:norm_deviations}
\end{figure}

Theorem \ref{thm3} guarantees that for any $r$ satisfying $r^* < r < 1$, system \eqref{eq:whole_system} synchronizes faster than $r^t$, where $r^* = 0.7821$.
Numerical results are shown in the figures:

\begin{itemize}
\item Figures \ref{fig:state1}, \ref{fig:state2}, \ref{fig:state3} show the evolution of the three state components of $\xi_1,\ldots,\xi_4$ relative to the average.
The simulation results are consistent with Theorem \ref{thm2}: the multi-agent system achieves synchronization.
\item Figure \ref{fig4} shows the deviation of $\xi_i$ from the average in semilogarithmic coordinates. 
The figure indicates that the deviation decays exponentially, implying that the multi-agent system achieves exponential synchronization.
Figures \ref{fig5}–\ref{fig7} show the ratio of this deviation to $0.9^t$, $0.8^t$, and $0.7^t$, respectively, also in semilogarithmic coordinates. 
Specifically:
\begin{itemize}
    \item Figure \ref{fig5} shows that the deviation decays faster than $0.9^t$;
    \item Figure \ref{fig6} shows that the deviation decays roughly as $0.8^t$;
    \item Figure \ref{fig7} shows that the deviation decays slower than $0.7^t$.
\end{itemize}
These simulation results are consistent with Theorem \ref{thm3}: 
the multi-agent system indeed achieves exponential synchronization, 
and the actual convergence rate is close to the theoretical estimate $r^*$.
\end{itemize}

\end{example}
\section{Conclusion and Future work}
This paper has studied asymptotic decoupling for a class of discrete-time LTV systems 
and applied the result to the autonomous synchronization of heterogeneous MASs.
As the application, Theorem \ref{thm2} provides sufficient conditions for synchronization, 
and Theorem \ref{thm3} gives an explicit estimate of the convergence rate, which numerical simulations show to be close to the actual convergence behavior.

Although serving as a significant improvement upon the results of \cite{5},
the proposed framework still relies on global knowledge of the graph topology and the average dynamics $S_\infty$.
Future work will focus on developing fully distributed protocols.

\bibliographystyle{IEEEtran}
\bibliography{references}

\end{document}